\documentclass[twocolumn,journal]{IEEEtran}
\usepackage[T1]{fontenc}
\usepackage[latin9]{inputenc}
\usepackage{amsmath}
\usepackage{amssymb}
\usepackage{graphicx}
\usepackage{subfigure}
\usepackage{commands11}
\usepackage{verbatim} 
 
\renewcommand{\dim}{\text{dim}\:} 
\newtheorem{definition}{Definition}
\newtheorem{proposition}{Proposition}
\newtheorem{lemma}{Lemma}
\newtheorem{thm}{Theorem}

\newtheorem{cor}{Corollary}

\newtheorem{remark}{Remark}

\begin{document}
%
\title{Entropy functions and determinant inequalities}

\author{Terence Chan, Dongning Guo, and Raymond Yeung}

\markboth{Journal of \LaTeX\ Class Files,~Vol.~6, No.~1, January~2007}%
{Shell \MakeLowercase{\textit{et al.}}: Bare Demo of IEEEtran.cls for Journals}
 
\maketitle
 
\IEEEpeerreviewmaketitle

\def\rgs{{\Upsilon_{s,n}}}
\def\rgv{{\Upsilon_{v,n}}}

\def\rd{{\Gamma}^{*}}
\def\rs{{\gamma_{s,n}^{*}}}
\def\rv{{\gamma^{*}_{v,n}}}
\def\X{{\set{X}}}
\def\cov{{\text{Cov}}}
\newcommand{\ed}[1]{{\Phi}^{#1}}

\renewcommand{\bar}[1]{ \overline{#1} }
\newcommand{\seq}[2]{{[#1 , \ldots, #2]}}

\newcommand{\nle}[1]{\stackrel{#1}{\le}}
\newcommand{\nge}[1]{\stackrel{#1}{\ge}}
\newcommand{\nequal}[1]{\stackrel{#1}{=}}
\def\con{\overline{\text{con}}}

\def\N{{\set{N}}}

\def\d{{\it d}}
\def\s{{\it s}}
\def\v{{\it v}}

\begin{abstract}
In this paper, we show that the characterisation of all determinant inequalities for 
$n \times n$ positive definite matrices is equivalent to determining the smallest closed and convex cone containing all entropy functions induced by $n$ scalar Gaussian random variables. We have obtained inner and outer bounds on the cone by using representable functions and entropic functions.  In particular, these bounds are tight and explicit for $n \le 3$, implying that determinant inequalities 
for $3 \times 3$ positive definite matrices are completely characterized by Shannon-type information inequalities.
\end{abstract}

\section{Introduction}
 
Let $n$ be a positive integer and denote the ground set by $\N=\{1,...,n\}$ throughout this paper. Suppose  $K$ is an $n\times n$ positive definite matrix. For any  subset $\alpha\subseteq \set{N}$, let $K_{\alpha}$ be the sub-matrix of $K$ obtained by removing those rows and columns of $K$ indexed by  $\set{N} \setminus \alpha$ and its determinant be denoted by $|K_{\alpha}|$. 
Note that when $\alpha$ is the empty set, we will simply define $K_{\alpha}$ as the scalar of value 1.
There are many determinant inequalities in the existing literature that involve only
the principle minors of the matrix.  These include  
\begin{enumerate}
\item  Hadamard inequality
\begin{align}
|K| \le \prod_{i=1}^{n} |K_{i}|
\end{align}

\item   Szasz inequality 

\begin{align}
\left( \prod_{\beta \subseteq \N: |\beta| = l} |K_{\beta}| \right)^{\frac{1}{{k-1
\choose l-1} }} 
\ge 
\left(
\prod_{\beta\subseteq \N: |\beta| = l+1} |K_{\beta}|
\right)^{  \frac{1}{{k-1 \choose l}} }
 \end{align}
for any $1\le l <k$.
 \end{enumerate}
 
As pointed out in \cite{CoverBook,Cover88} and to be illustrated in Section \ref{sec:framework},  many of such determinant inequalities (including the above two inequalities) can be proved via an information-theoretic approach.  Despite that  many determinant inequalities can be found in this approach, a complete characterisation of all determinant inequalities is still missing. In this paper, we aim to understand 
 determinant inequalities by using the information inequality framework proposed in \cite{R.W.Yeung1997A-framework}. 

\section{Information inequality framework}\label{sec:framework}

The framework proposed in \cite{R.W.Yeung1997A-framework} provides a geometric approach to understanding  information inequalities.\footnote{See \cite[Ch.~13-16]{Yeung08} for a comprehensive treatment.} Its idea will be illustrated shortly. 

\begin{definition}[Rank functions]
 A \emph{rank function} over the ground set $\set{N}$ is a real-valued function defined on all subsets of $\set{N}$. 
 The \emph{rank function space} over the ground set $\set N$, denoted by ${\mathbb R}^{2^{n}}$, is the set of all rank functions over $\set{N}$. 
\end{definition}

As usual,  ${\mathbb R}^{2^{n}}$ will be treated as a $2^{n}$-dimensional Euclidean space, so that  concepts such as metric and limits can be defined accordingly.

\begin{definition}[Entropic functions]
Let $g$ be a rank function over $\set{N}$. Then $g$ is called 
\emph{entropic} if there exists  
a set of discrete random variables $\{ X_{i}, i\in \set{N}\}$ such that 
$g(\alpha)$ is   the Shannon entropy\footnote{All logarithms used in the paper is in the base 2.}  $H(X_i, i\in {\alpha} )$, or  $H(X_{\alpha} )$ for short, for all $ \alpha\subseteq \set{N}$. 

On the other hand, if  $\{ X_{i}, i\in \set{N}\}$ is a set of continuous scalar random variables such that $g(\alpha)$ is the differential entropy $h(X_\alpha)$ for all $ \alpha\subseteq \set{N}$, then  $g$ is called \emph{s-entropic}.
%
%
%
\end{definition}


\begin{definition}[Entropic regions]
Consider any nonempty finite ground set $\N$. Define the following ``entropy regions'':
\begin{align}
\rd_{n}  &= \{g \in {\mathbb R}^{2^n} : g \text{ is  entropic}\} \\
\rs  &=  \{g \in {\mathbb R}^{2^n} : g \text{ is  \s-entropic}\}.
\end{align} 
 \end{definition}

Understanding the above entropic regions is one of the most fundamental problems in information theory. It is equivalent to determining the set of all information inequalities~\cite{R.W.Yeung1997A-framework}.


In this paper, we will use the following notation. For any subset $\set{S} \subseteq {\mathbb R}^{2^n}$, $\mathbb{W}(\set{S})$ is defined as the set of all
rank functions $g^{*}$ such that $g^{*}=c \cdot g$ for some  $c>0$ and $g\in \set{S}$.
The closure of $\mathbb{W}(\set{S})$ will be denoted by $\bar{\mathbb{W}}(\set{S})$.
 Finally, the smallest closed and convex cone containing $\set{S}$ will be denoted by $\con(\set{S})$. Clearly, 
\begin{align}
\set{S} \subseteq \mathbb{W}(\set{S}) \subseteq  \bar{\mathbb{W}}(\set{S}) \subseteq \con(\set{S}).
\end{align}

\begin{thm}[Geometric framework~\cite{R.W.Yeung1997A-framework}] \label{thm:yeungframework}
A linear information inequality 
\[
\sum_{\alpha \subseteq \N} c_{\alpha} H(X_\alpha) \ge 0
\]
is valid for all  discrete random variables $\{X_{1},\ldots, X_{n}\}$ 
 if and only if for all $g\in\rd_n$
\[
\sum_{\alpha \subseteq \N} c_{\alpha} g(\alpha) \ge 0.
\]
\end{thm} 

By Theorem \ref{thm:yeungframework}, characterising the set of all valid information inequalities is thus equivalent to characterising 
the set $\rd_n$. 
Similar results can be obtained for the set $\rs $. 
In the following, we will extend this geometric framework to study 
determinant inequalities.

\def\detfn{{\Psi_{n}}}

\begin{definition}[Log-determinant function]
A rank function $g$ over $\set{N}$ is called \emph{log-determinant} if there exists  an $n\times n$ positive definite matrix $K$ such that 
\begin{align}
g(\alpha) = \log |K_{\alpha}|
\end{align}
for all $ \alpha \subseteq \N$. 
\end{definition}

Let $\detfn  $ be the set of all log-determinant functions over $\set{N}$. 
Then, we have the following theorem.

\begin{thm}\label{thm:2}
Let $\{c_{\alpha}, \alpha \subseteq \N \}$ be any real numbers. The determinant inequality
\begin{align}
\prod_{\alpha\subseteq \N} |K_{\alpha}|^{c_{\alpha}} \ge 1 \label{eq:thm2}
\end{align}
holds for all positive definite matrix $K$ if and only if 
\begin{align}
\sum_{\alpha\subseteq \N} c_{\alpha } g(\alpha)  \ge 0
\end{align}
for all $g \in  \con(\detfn) $.
\end{thm}
\begin{proof}
By taking logarithm on both sides of the inequality,   \eqref{eq:thm2} is equivalent to that 
\begin{align}
\sum_{\alpha\subseteq \N} {c_{\alpha}} \log |K_{\alpha}| \ge 0  \label{eq:thm2a}
\end{align}
for all positive definite matrix $K$. 
As \eqref{eq:thm2a} is  a linear inequality, it is satisfied by all $g\in \detfn$ if and only if it is satisfied by all $g \in  \con(\detfn) $. The theorem then follows.
\end{proof}

In other words,   the characterisation of the set of all determinant inequalities is equivalent to determining the set $\con(\detfn)$. In the rest of the paper, we will obtain inner and outer bounds on $\con(\detfn) $. 

To achieve our goal, we will take an information theoretic approach~\cite{Cover88}. The idea is very simple:  Let $\{X_{1}, \ldots, X_{n}\}$ be 
a set of scalar Gaussian random variables whose covariance matrix is equal to $  ( 1/ {2\pi e}) K$. Then the differential entropy of $X_\alpha$ is given by   
\begin{align}\label{eq:7}
h(X_\alpha) = \frac{1}{2} \log |K_{\alpha}| .
\end{align}
 
\begin{definition}[Scalar Gaussian function]\label{df:Gaussian}
A function $g \in {\mathbb R}^{2^n}$ is called \emph{s-Gaussian}  if 
there exists scalar Gaussian variables  $\{ X_{1} , \ldots, X_{n} \}$ where 
\begin{align}
g(\alpha) = h(X_\alpha)
\end{align}
for all $\alpha \subseteq \N$. 
\end{definition}

From \eqref{eq:7},  a rank function $g$ is log-determinant if and only if $\frac{1}{2}g$ is \s-Gaussian.
Let $\rgs$  be the set of all \s-Gaussian functions. Then 
\[ \con(\detfn) =  \con(\rgs).\]
Consequently, we have the following theorem.

\begin{thm}
The determinant inequality
\[
\prod_{\alpha\subseteq \N} |K_{\alpha}|^{c_{\alpha}} \ge 1
\]
holds for all positive definite matrix $K$ if and only if 
\[
\sum_{\alpha\subseteq \N} c_{\alpha } h(X_\alpha)  \ge 0
\]
for all scalar Gaussian variables  $\{ X_{1} , \ldots, X_{n} \}$.
\end{thm}

In fact, the Hadamard inequality and  Szasz inequality are respectively  the counterparts of the following basic information inequalities\footnote{Han's inequality was originally proved for discrete random variables. However, by using the same proving technique, it can also be proved to hold for all continuous random variables~\cite{CoverBook}. Alternative,  its validity also follows from \cite{Chan2003Balanced}: If a balanced information inequality (including Han's inequality) holds for all discrete random variables, then its ``continuous counterpart'' (i.e., the inequality by replacing discrete entropies with differential entropies) also holds for all continuous random variables.}~\cite{Han78}
\begin{align}
\sum_{i=1}^{n} h(X_{i})  & \ge h(X_{1}, \ldots, X_{n}) \\
\frac{1}{{{k \choose l}}} \sum_{\beta \subseteq \N: |\beta| = l}
\frac{h(Y_\beta)}{l}
&\ge
\frac{1}{{{k \choose l+1}}} \sum_{\beta \subseteq \N: |\beta| = l+1}
\frac{h(Y_\beta)}{l+1}.
\end{align}

In the following sections, we will obtain inner and outer bounds on the set $\con(\rgs)$.
The following corollaries of Theorem \ref{thm:2} show 
how these bounds can be used for proving or disproving a determinant inequality.

\begin{cor}[Proving an inequality] 
Suppose $ \set S $ contains $  \con(\rgs)$ as a subset.
The determinant inequality \eqref{eq:thm2}
 holds for all positive definite matrix $K$ if 
\[
\sum_{\alpha\subseteq \N} c_{\alpha } g(\alpha)  \ge 0
\]
for all $g \in \set S $.
\end{cor}

Therefore, any explicit outer bound on  $  \con(\rgs)$ can lead to the discovery of 
new determinant inequalities. On the other hand, an inner bound on 
$  \con(\rgs)$ can be used for disproving a determinant inequality.

\begin{cor}[Disproving an  inequality] 
Suppose $ {\set T} \subseteq \con(\rgs) $. The determinant inequality \eqref{eq:thm2}
 does not hold for all positive definite matrices  if there exists $g\in \set T$ such that 
\[
\sum_{\alpha\subseteq \N} c_{\alpha } g(\alpha) < 0.
\] 
\end{cor}


%

\section{An inner bound and an outer bound}

As discussed earlier,  log-determinant functions are essentially the same as \s-Gaussian functions. 
Our objective is thus  to characterise $\con(\rgs)$,  or at least to understand its basic properties.
%
%
Since scalar Gaussian random variables are continuous scalar random variables, the next lemma follows immediately from the definition.
\begin{lemma}[Outer bound]\label{lemma:outerbd}
\begin{align}
\rgs  \subseteq \rs ,
\end{align}
and consequently, 
\begin{align}\label{eq:firstout}
\con(\rgs  ) \subseteq \con(\rs ) .
\end{align}
\end{lemma}

It is well known that ${\bar\Gamma}_{n}^{*} $ (i.e., the closure of $\Gamma_{n}^{*}$) is a closed and convex cone~\cite{R.W.Yeung1997A-framework}. It was established in \cite{Chan2003Balanced} that 
\begin{align}\label{5}
\con( \rs)    = \con({\bar\Gamma}^*_n , \phi^n_{1} , \ldots, \phi^n_{n})  
\end{align}
where
\[
\phi^n_{i}(\alpha) = 
\begin{cases}
-1 & \text{ if } i\in\alpha \\
0 & \text{ otherwise. }
\end{cases}
\]

In the following, we prove an inner bound on $\con(\rgs )$ by using representable functions.

\begin{definition}[$s$-representable function]
A rank function $g$ over $\N$ is called \emph{s-representable} if  there exists  real-valued vectors (of the same length) $\{ A_{1} , \ldots ,  A_{n}\} $ such that for all $\alpha\subseteq \N$,
\[
g(\alpha ) = \dim { \langle  A_{i} , i\in\alpha\rangle }.
\]
In other words, $g(\alpha )$ is the maximum number of independent vectors in the set $\{ A_{i} , i\in\alpha  \}$.
\end{definition}

\begin{thm}[Inner bound]\label{thm:repisgaussian}
If $g$ is \s-representable, then 
\[
g\in{\bar {\mathbb W}} (\rgs).
\]
\end{thm}
\begin{proof}
Suppose the length of each row vector $A_{i}$ is  $k$.
Let 
\[
\{ W_{1}, \ldots, W_{k}, V_{1}, \ldots, V_{n} \}
\]
 be a set of independent standard Gaussian random variables. Therefore, its covariance matrix  is the $(n+k) \times (n+k)$ identity matrix. Let $c>0$.  For each $i=1,\ldots , n$, define a real-valued continuous random variable as follows
\[
X_{i} \triangleq \frac{1}{\sqrt{c}}A_{i} [ W_{1} , \ldots , W_{k} ]^{\top} + V_{i}.
\]
Let ${\bf X} = [X_{1}, \ldots, X_{n}]^{\top}$. Then 
\[
{\bf X} = \frac{1}{\sqrt{c}}A  [ W_{1} , \ldots , W_{k} ]^{\top} + {\bf V} 
\]
where $A$ is an $n \times k$ matrix whose $i^{th}$ row is $A_{i}$ and 
\[
{\bf V} = [V_{1}, \ldots, V_{n}]^{\top}.
\]

Since $X_{i}$ is zero-mean,  
\begin{align*}
\cov ({\bf X}) &= E[ {\bf X} {\bf X}^{\top}]  \\
&= \frac{1}{c} E\left[ A  [ W_{1} , \ldots , W_{k} ]^{\top}  [ W_{1} , \ldots , W_{k} ] A^{\top} \right] +  {\bf I} \\
&= \frac{1}{c}  A A^{\top}  +  {\bf I} .
\end{align*}
Consequently, 
\begin{align}
\det (\cov ({\bf X})) = \det \left(  \frac{1}{c} D + {\bf I} \right) 
\end{align}
where $D$ is the diagonal matrix obtained by using singular-value decomposition (SVD) over $ A A^{\top}$. 
Let $d_{1} \ge d_{2} \ge \cdots \ge d_{n} \ge 0$ be the diagonal entries of $D$ and $r$ be the rank of the matrix $ A A^{\top}$ (or equivalently, the rank of $A$).  Hence, $d_{i} >0$ if and only if $i \le r$. Then 
\begin{align}
\det (\cov ({\bf X})) = \prod_{i=1}^{r}  \left(  \frac{d_{i}}{c}  + 1  \right). 
\end{align}
It is easy to see that 
\begin{align}
\lim_{c \to  0} \frac{h(X_{1}, \ldots, X_{n})}{\frac{1}{2} \log 1/c} & =  \lim_{c \to 0} \frac{ \frac{1}{2} \log \left(  (2 \pi e)^{n} \det (\cov ({\bf X})) \right) }{\frac{1}{2} \log 1/c}\\
 & =  \lim_{c \to 0} \frac{  \log \left(   \det (\cov ({\bf X})) \right) }{  \log 1/c}\\
  & =  \lim_{c \to 0} \frac{ \sum_{i=1}^{r}   \log   \left(  \frac{d_{i}}{c}  + 1  \right) }{  \log 1/c}\\
&= r.
\end{align}

Similarly, for any $\alpha \subseteq \{1, \ldots, n\}$, we can prove that 
\[
\lim_{c\to 0} \frac{h(X_\alpha) }{\frac{1}{2} \log 1/c} = \dim \langle  A_{i}, i\in\alpha \rangle = g(\alpha).
\]
Thus, $g\in {\bar {\mathbb W}} (\rgs)$ and the theorem is proved. 
\end{proof}

\begin{lemma}\label{lemma:5}
Let $\{X_{1} , \ldots, X_{n} \}$ be a set of scalar jointly continuous random variables with 
differential entropy function $g$. 
For any $c_{1}, \ldots, c_{n } > 0$, define the set of random variables
$\{Y_{1} , \ldots, Y_{n} \}$ by 
\[
Y_{i} = X_{i} /c_{i}  , \: \forall i\in\N ,
\]
and let $g^{*}$ be the differential entropy function of $\{Y_{1} , \ldots, Y_{n} \}$. 
 Then 
\begin{align}\label{eq:41}
g^{*}(\alpha) & = g(\alpha) + \sum_{i\in\alpha} \log c_{i} \\
& = g(\alpha) - \sum_{i\in\N} (\log c_{i} ) \phi^n_{i}(\alpha) 
\end{align}
for all $\alpha \subseteq \N$. 
Consequently, if $g$ is \s-Gaussian, then so is $g^{*}$.
\end{lemma}
\begin{proof} 
Let $f_{X_{1},\ldots, X_{n}}$ and $f_{Y_{1},\ldots, Y_{n}}$ be respectively the probability density functions (pdfs) of  $\{X_{1} , \ldots, X_{n} \}$ and $\{Y_{1} , \ldots, Y_{n} \}$. Then  
\begin{multline}
f_{Y_{1},\ldots, Y_{n}} (y_{1}, \ldots, y_{n}) \\ 
= \left( \prod_{i=1}^{n}  c_{i} \right) f_{X_{1},\ldots, X_{n}}(c_{1}y_{1} , \ldots, c_{n}y_{n}),
\end{multline}
and \eqref{eq:41} can be directly verified. 
\end{proof}

\begin{cor}\label{cor:3}
\begin{multline}
\con(\Omega_{s,n}, \phi^n_{1} , \ldots, \phi^n_{n}) \subseteq   \con(\rgs ) \subseteq  \con(\rs) \\ =\con({\bar\Gamma}^*_n , \phi^n_{1} , \ldots, \phi^n_{n} )
\end{multline}
where $\Omega_{s,n}$ is the set of all \s-representable functions.
\end{cor}
\begin{proof}
A direct consequence of Lemmas \ref{lemma:outerbd} and \ref{lemma:5},  Theorem \ref{thm:repisgaussian} and \eqref{5}.
\end{proof}

\begin{proposition}[Tightness of inner and outer bounds]\label{prop:Nis3}
For $n \le  3$, 
\begin{multline}
\con(\Omega_{s,n}, \phi^n_{1} , \ldots, \phi^n_{n}) =  \con(\rgs ) \\ =  \con(\rs)  =\con({\bar\Gamma}^*_n , \phi^n_{1} , \ldots, \phi^n_{n} ).
\end{multline}
\end{proposition}
\begin{proof}
By Corollary \ref{cor:3}, to prove the proposition, it suffices to prove that for $n \le 3$,
\begin{align}\label{eq:28}
\con({\bar\Gamma}^*_n ) \subseteq \con(\Omega_{s,n}).
\end{align}
In \cite{Zhang1998On-the-characterization}, the cone ${\bar\Gamma}^*_{n}$ (when $n\le 3$) was explicitly determined by identifying the set of extreme vectors 
of the cone.  It can be proved that all the extreme vectors are \s-representable\footnote{In \cite{Zhang1998On-the-characterization}, the extreme vectors are proved to be  representable with respect to a finite field. However, it can be verified easily that they are also \s-representable with respect to the real field $\mathbb R$.} and hence is a subset of 
$\con(\Omega_{s,n})$. Consequently, \eqref{eq:28} holds and the proposition follows.
\end{proof}

Proposition \ref{prop:Nis3} does not hold when $n \ge 4$. 
In fact, $\con(\Omega_{s,n}, \phi^n_{1} , \ldots, \phi^n_{n})$ is in general  a proper subset of $\con(\rgs)$ when $n \ge 4$. 
In \cite{ingleton71}, it was proved that all \s-representable functions satisfy the Ingleton inequalities.
It can also be directly verified that all the functions $\phi^n_{i}$ also satisfy the Ingleton inequalities. 
Therefore, all the functions in   $\con(\Omega_{s,n}, \phi^n_{1} , \ldots, \phi^n_{n})$ also satisfy the
Ingleton inequalities. However, in \cite{Hassibi2008The-entropy}, it was proved that there exists $g\in \rgs$ for $n=4$  that 
 violates the the Ingleton inequality. Thus, $\con(\Omega_{s,n}, \phi^n_{1} , \ldots, \phi^n_{n})$ is indeed a proper subset of $\con(\rgs)$.


\section{Another outer bound}

By definition, the set $\con(\detfn )$ (which is the focus of our interest) is close under addition. However, this is not necessarily true for $\detfn$. In fact, $\bar{\mathbb{W}}(\detfn)$ is not necessarily equal to $\con(\detfn)$.

In the previous section, we showed that the set $\detfn$ is essentially equivalent to the set of \s-Gaussian functions, defined via sets of scalar Gaussian random variables. It turns out that, if we relax the constraint by allowing the Gaussian random variables to be vectors, instead of scalars, we will obtain an outer bound for $\detfn$ and also  $\con(\detfn)$.
 
\begin{definition}[Vector Gaussian function]
A function $g \in {\mathbb R}^{2^n}$ is called \emph{v-Gaussian}  if 
there exists $n$   Gaussian random vectors  $\{ X_{1} , \ldots, X_{n} \}$ 
such that 
\begin{align}
g(\alpha) = h(X_\alpha)
\end{align}
for all $\alpha \subseteq \N$. 
\end{definition}

\begin{lemma} \label{lemma:3}
$\con(\rgv) = \bar{\mathbb{W}}(\rgv)$.
\end{lemma}
\begin{proof}
It is clear from the definition that  
${\bar {\mathbb W}} (\rgv) \subseteq {\con(\rgv )} $. Now, consider positive integers $k, \ell_{1}, \ell_{2}$ and  $g_{1}, g_{2} \in \rgv$. It is easy to see that 
\[
\ell_{1} g_{1} + \ell_{2} g_{2} \in \rgv.
\]
Hence,
\[
\frac{\ell_{1}}{k} g_{1} + \frac{\ell_{2}}{k} g_{2}   \in \mathbb W(\rgv).
\]
 Since $k, \ell_{1}, \ell_{2}$ are arbitrary positive integers,   for any positive numbers $c_{1}, c_{2} > 0$, 
\[
c_{1} g_{1} + c_{2 }g_{2} \in  {\bar {\mathbb W}} (\rgv)
\]
 and the lemma follows. 
\end{proof}

\begin{thm}[Another outer bound]
\begin{align}\label{eq:secondout}
 \con(\rgs)  \subseteq  {\bar {\mathbb W}} (\rgv). 
\end{align}
\end{thm}
\begin{proof}
A direct consequence of that $\rgs \subseteq \rgv$ and Lemma \ref{lemma:3}.
\end{proof}

So far,  we have established two outer bounds \eqref{eq:firstout} and \eqref{eq:secondout} for $\con(\rgs)$. 
In the following, we will prove that \eqref{eq:secondout} is in fact a tighter one.

\begin{definition}
A rank function $g$ is called \emph{v-entropic} 
if there exists a set of random vectors $\{X_{1}, \ldots, X_{n}  \}$, not necessarily of the same length, such that 
\[
g(\alpha) = h(X_\alpha).
\]
Also, let 
\begin{align}
\rv(\set{N}) &=  \{g \in {\mathbb R}^{2^n} : g \text{ is  \v-entropic}\}.
\end{align} 
\end{definition}

Clearly, ${\bar {\mathbb W}} (\rgv) = \con(\rgv) \subseteq  \con(\rv)$. Thus, 
\[
\con(\rgs) \subseteq {\bar {\mathbb W}} (\rgv)  \subseteq  \con(\rv).
\]
To show that  \eqref{eq:secondout} is tighter, it suffices to prove the following result. 
 

\begin{thm}\label{thm:main}
$\overline{\rv}  = \overline{\rs}  = \con({\bar\Gamma}_{n}^*, \phi^n_{1} , \ldots, \phi^n_{n})$.
\end{thm}

  Theorem~\ref{thm:main} basically states that replacing the
  real-valued random variables $X_i$ in the vector $\X$ by random
  vectors does not enlarge the closure of the space of differential entropy
  vectors.  The discrete counterpart of this result is trivial,
  because as far as the probability masses and the entropy are concerned, a
  discrete random vector can be replaced by a scalar discrete random
  variable.  However, in the continuous domain, it is not clear how a
  probability density function on $\reals^2$ or more generally
  $\reals^m$ can be mapped to a pdf on $\reals$ without changing the
  entropies.  In particular, there does not exist a continuous mapping
  from $\reals^2$ to $\reals$~\cite{Wiboonton2010Bijections}.

%

The proof of Theorem~\ref{thm:main} exploits the relationship between the differential entropy of a continuous vector and the entropy of a discrete vector obtained through quantisation.  Moreover, the entropy of the discrete random variable is equal to the differential entropy of a continuous random variable with piece-wise constant pdf.
Given the $n$-tuple $\Z$ whose entries are vectors, we ``quantise'' $\Z$ by a discrete vector and then construct a continuous vector with $n$ scalar entries whose entropy vector arbitrarily approximates that of $\Z$.
Before we prove the theorem, we  need several intermediate supporting results.

\def\bX{{\bf X}}

\begin{lemma}[Closeness in addition]\label{lemma:addition}
If $g_{1} $ and $g_{2}$ are \v-entropic (or entropic) functions over $\N$, then 
their sum $g_{1} + g_{2} $ is also \v-entropic (or entropic).
\end{lemma}
\begin{proof}
Direct verification.
\end{proof}

\begin{proposition}\label{prop:claim}
If $g^{*}\in\rv$, then for any $c>0$, 
$c \cdot g^{*} \in \bar\rv$. 
\end{proposition}
\begin{proof}
Let  $\bX=(X_1,\dots,X_n)$  be a real-valued random vector with
a probability density function.  For any positive integer $j$, let $\bX^{(1)}, \dots, \bX^{(j)}$   be $j$ independent replicas of $\bX$ (by a replica we mean a random
  object with identical distribution). 
Similarly, let $\U=(U_1,\dots,U_n)$ be a real-valued random vector such that 
$U_1,\dots,U_n$ are mutually independent and each of them is uniformly distributed 
on the interval $[0,1]$. Again, for any positive integer $j$, let $\U^{(1)}, \dots, \U^{(j)}$
  be $j$ independent replicas of $\U$. It is easy to see that the joint density function of $\U^{(1)}, \dots, \U^{(j)}$ is uniform on a  hypercube with unit volume and hence has zero differential entropy.

Consider any $c>0$. Let 
$T$ be a binary random variable such that  
\[
\Probk{T=1}= c/j \text{ and } \Probk{T=0}= 1-c/j
\]
where $j$ is a positive integer.
Assume that $T$ is independent of 
 \[
( \bX^{(1)}, \U^{(1)} \ldots, \bX^{(j)} , \U^{(j)}).
 \]
Let $\Z = ( Z_{1}, \ldots, Z_{n})$ where each $Z_{i}$ is a random vector of length $j$ such that for 
any $i=1,\ldots, n$, 
 \begin{align}
  \label{eq:17}
   Z_i =
   \begin{cases}
     (U_i^{(1)}, \dots, U_i^{(j)} ) &
    \quad \text{if } T=0\\
     (X_i^{(1)}, \dots, X_i^{(j)}) &
     \quad \text{otherwise. }  
 \end{cases}
 \end{align}
 $\Z$ is evidently continuous with a  pdf, which is a mixture of two 
 pdfs induced by that of $\bX$ and $\U$. 
For any $\alpha \subseteq \mathcal N$, we can directly verify that 
\begin{align}
h(Z_\alpha | T=0) & = h(U_\alpha^{(1)}, \dots, U_\alpha^{(j)}  ) \\
& = 0
\end{align}
and
\begin{align}
h(Z_\alpha | T=1) & = h(X_\alpha^{(1)}, \dots, X_\alpha^{(j)}  ) \\
& = j h(X_\alpha).
\end{align}
Consequently, 
\begin{align}
h(Z_\alpha | T)  & = c h(X_\alpha).
\end{align}

Hence, 
\begin{align}
 c h(X_\alpha)  & =\lim_{j\to\infty} h(Z_\alpha | T ) \\
  & \le \lim_{j\to\infty} h(Z_\alpha) \\
 &\le  \lim_{j\to\infty} h(Z_\alpha | T ) + h_{b}(c/j) \\ 
&= c h(X_\alpha),
\end{align}
where 
$h_{b}(x) $ is the entropy of a binary random variable with probabilities $x$ and $1-x$.
Thus,  $\lim_{j\to\infty} h(Z_\alpha) = c h(X_\alpha)$.
Let $g^{j}$ and $g^{*}$ be respectively the entropy function induced by $\{Z_{1} , \ldots, Z_{n}  \}$ and $\{X_{1} , \ldots, X_{n}  \}$. Then
$g^{j}$ is $v$-entropic by definition and  
\[
\lim_{j\to\infty} g^{j} = c\cdot g^{*}.
\]
Hence, $c \cdot g^{*} \in \bar\rv$ for all $c>0$ and our proposition follows.
\end{proof}

\begin{proposition}\label{prop:cone}
$\overline{\rv}$ is a closed and convex cone.  
\end{proposition}
\begin{IEEEproof}
For any $r\in \bar\rv$, by definition, there exists a sequence of $v$-entropic  functions $\{r^{i} \}_{i=1}^{\infty}$ such that 
\[
\lim_{i\to\infty} r^{i} = r.
\]
Thus,  for any $c>0$, 
\[
\lim_{i\to\infty} c \cdot r^{i} = c \cdot r.
\]
Then, by Proposition \ref{prop:claim}, $c \cdot r^{i} \in \bar\rv$ and consequently, $c \cdot r \in \bar\rv$.

Consider any $g_{1}^{*},g_{2}^{*} \in\bar\rv$, and $c_{1}, c_{2} > 0$. Since  
\[
c_{1} \cdot g_{1}^{*} \text{ and }  c_{2} \cdot g_{2}^{*} \in \bar\rv,
\]
  there exists sequences of $v$-entropic functions $\{r^{i}_{1}\}_{i=1}^{\infty}$ and $\{r^{i}_{2}\}_{i=1}^{\infty}$ such that 
\[
\lim_{i\to\infty} r^{i}_{\ell} = c_{\ell} \cdot g^{*}_{\ell}.
\]
By Lemma \ref{lemma:addition}, $r^{i}_{1} + r^{i}_{2}$ is also $v$-entropic. Thus, 
\[
c_{1} \cdot g_{1}^{*} +  c_{2} \cdot g_{2}^{*} \in \bar\rv.
\]
The proposition is proved.
\end{IEEEproof}

%

\begin{definition}[$m$-Quantization]
  Given $m>0$, let the $m$-quantization of any real number $x$ be denoted as:
 \begin{align}    \label{eq:Xm}
   [x]_m = \frac{ \lfloor m x \rfloor } m
\end{align}
where $\lfloor t \rfloor$ denotes the largest integer not exceeding $t$.
Similarly, let the $m$-quantization of a real vector
$\x=(x_1,\dots,x_n)$ be the element-wise $m$-quantization of the
vector, denoted by $[\x]_m$, i.e.,
\begin{align}
  \label{eq:13}
  [\x]_m = ( [x_1]_m, \dots, [x_n]_m )\ .
\end{align}
\end{definition}

Evidently, $[x]_m$ can only take values from the set 
\begin{align}\label{eq:supportXm}
\left\{0, \pm\frac1m, \pm\frac2m, \dots\right\}.
\end{align}
Hence for every real-valued random variable $X$, 
$[X]_m $ is a discrete random variable taking value in the set \eqref{eq:supportXm}. By definition,  
\begin{align}  \label{eq:4}
 \sum_{i\in\integers} \Probk{ [X]_m = \frac{i}m } = 1.
\end{align}

%
%
%
%

\begin{proposition}[Renyi \cite{Ren70}] \label{pr:Renyi}
If $X$ is a real-valued random vector of dimension $n$ with a
probability density function, then
\begin{align}    \label{eq:Renyi}
  \lim_{m\to\infty} H([X]_m) - n \log m = h(X) \ .
\end{align}
\end{proposition}

  Under the assumption that the pdf of a random variable $X$ is
  Riemann-integrable, Proposition~\ref{pr:Renyi} is established
  in~\cite{CovTho06} by treating $H([X]_m) - n \log m$ as the
  approximation of the Riemann integration of $-\int f_X(x) \log f_X(x)
  \diff x$.
  It is nontrivial to establish the result in general, where the pdf
  is not necessarily Rieman-integrable.  An example of such a pdf can be
  defined by using the Smith-Volterra-Cantor set.
  Nonetheless~\eqref{eq:Renyi} can be shown to hold using the
  Lebesgue convergence theorem along with some truncation
  arguments~\cite{Ren70}. 
 
\def\X{{\bf X}}

\begin{lemma}\label{lemma:diss}
Let $\{X_{1}, \ldots, X_{n} \}$ be a set of discrete random variables such that its entropy function is $g$.
For any positive numbers $c_{1}, \ldots, c_{n}$, let $g^{*}$ be defined as  
\[
g^{*}(\alpha) = g(\alpha) - \sum_{i\in\alpha} \log c_{i}.
\]
Then $g^{*}$ is \s-entropic.
\end{lemma} 
\begin{IEEEproof}   
As  $X_{i}$ is discrete, we may assume without loss of generality that the sample space of $X_{i}$ is the set of integers $\mathbb Z$. Let $p(x_{1}, \ldots, x_{n})$ be the probability mass function of  $\{X_{1}, \ldots, X_{n} \}$. Construct a set of continuous scalar random variables $\{Y_{1}, \ldots, Y_{n} \}$ whose probability density function is defined as follows:
\[
f_{Y_{1},\ldots,Y_{n}}(y_{1}, \ldots, y_{m}) \triangleq \left(\prod_{i=1}^{n} c_{i}\right) p(\lfloor c_{1}y_{1}\rfloor , \ldots, \lfloor c_{n}y_{n}\rfloor).
\] 
It can then be directly verified that 
\[
h(Y_\alpha) = H(X_\alpha)-  \sum_{i=1}^{n} \log c_{i}, \quad\forall \alpha \subseteq \set{N} .
\]
Consequently, $g^{*}$ is  \s-entropic.
\end{IEEEproof}

\begin{IEEEproof}[Proof of Theorem~\ref{thm:main}]
Clearly, $\bar\rs \subseteq \bar\rv$. We will now prove that $\bar\rv \subseteq \bar\rs$.
  Let $\Z=(Z_1,\dots,Z_n)$ consist of $n$ random vectors, where 
  \[
  Z_i = (Z_{i,1},\dots,Z_{i,k_i}).
  \]
    Let us define the $m$-quantization of
  $Z_{i}$, denoted as $[Z_{i}]_m$, be the element-wise $m$-quantization of
  $Z_{i}$, i.e., it consists of $[Z_{i,j}]_m$ for $j=1,\ldots, k_{i}$.  By Proposition~\ref{pr:Renyi},
  \begin{align}
    \label{eq:9}
    \lim_{m\to\infty} \left[ H([Z_{i}]_{m}, i\in\alpha) -  \left(\sum_{i \in \alpha } k_i \right)\log m \right] =  h(Z_\alpha) .
  \end{align}

Let $g^{*}, r^{m},g^{m} \in {\mathbb R}^{2^n}$ be such that 
\begin{align}
g^{*}(\alpha) & = h(Z_\alpha) \\
r^{m}(\alpha) & = H([Z_{i}]_{m}, i\in\alpha) \\
g^{m}(\alpha) &= r^{m}(\alpha) - \left(\sum_{i \in \alpha } k_i \right)\log m  .
\end{align}
By \eqref{eq:9}, $\lim_{m\to\infty} g^{m} = g^{*}$.
Also,  since $r^{m} \in \rd_{n}$, $g^{m}\in \rs$ by Lemma \ref{lemma:diss}.
Consequently, $g^{*}\in \bar\rs$. We have thus proved that $\rv \subseteq \bar\rs$ and as a result, $\bar\rv = \bar\rs$. 
Finally, by Proposition \ref{prop:cone}, $\bar\rv$ is a closed and convex cone and is equal to 
$\con(\rs)$.  Then by \eqref{5}, 
\begin{align}
\bar\rv    = \con({\bar\Gamma}^*_n , \phi^n_{1} , \ldots, \phi^n_{n}) . 
\end{align}
The theorem is proved.
\end{IEEEproof}

In Theorem \ref{thm:repisgaussian}, we have constructed an inner bound for $\con (\rgs)$ by using \s-representable functions. The same trick can also be used for constructing an inner bound for the set ${\bar {\mathbb W}} (\rgv)$.

\begin{definition}
A rank function $g$ over $\N$ is called \emph{v-representable} if for $i = 1,\ldots, n$, there exists a set of real-valued vectors (of the same length) $\{ A_{i,1} , \ldots  A_{i,k_{i}}\} $ such that for all $\alpha\subseteq \N$,
\[
g(\alpha ) = \dim { \langle  A_{i,j} , i\in\alpha, j=1,\ldots, k_{i}\rangle }.
\]
\end{definition}

The following theorem is a  counterpart of Theorem \ref{thm:repisgaussian}. The proving technique is the same as before. We will omit the proof for brevity. 

\begin{thm}[Inner bound on ${\bar {\mathbb W}} (\rgv)$]\label{thm:vrepisgaussian}
Suppose that  $g$ is \v-representable, then $g\in {\bar {\mathbb W}} (\rgv)$ . 
\end{thm}

Theorem \ref{thm:vrepisgaussian} is of great interest.  Characterising the set of \v-representable functions have been a very important problem in linear algebra and information theory. It is also extremely difficult. For many years, it is  only known that \v-representable functions are polymatroidal and satisfies the Ingleton inequalities~\cite{Guille2009The-minimal,ingleton71}. The set of representable functions is only known when $n \le 4$.
However, there were some recent breakthrough in this areas. In \cite{Dougherty2009Linear,Kinser2010New-inequalities}, many new subspace rank inequalities which are required to be satisfied by representable functions are discovered.  In particular, via a computer-assisted mechanical approach, the set of all representable functions when $n\le 5$ has been completely characterised. 
Interesting properties about the set of \v-representable functions were also obtained \cite{Chan2011Truncation}. Theorems \ref{thm:repisgaussian} and \ref{thm:vrepisgaussian} thus opens a new door to exploit results obtained about representable functions to characterise the set of Gaussian functions.

\begin{cor}[Inner bound on ${\bar {\mathbb W}} (\rgv)$]
\[
\con( \Omega_{v,n}, \phi^n_{1} , \ldots, \phi^n_{n} ) \subseteq {\bar {\mathbb W}} (\rgv)
\]
where $\Omega_{v,n}$ is the set of all \v-representable functions.
\end{cor}

\begin{remark}
While 
\[
\con( \Omega_{s,n}, \phi^n_{1} , \ldots, \phi^n_{n} ) \subseteq \con (\rgs),
\]
it is still an open question whether 
\[
\con( \Omega_{v,n}, \phi^n_{1} , \ldots, \phi^n_{n} ) \subseteq \con (\rgs)
\]
or not.
\end{remark}

We will end this section with a discussion of a related concept in a recent work~\cite{Hassibi2008The-entropy}. 
Gaussian rank functions were   studied in \cite{Hassibi2008The-entropy}. However, their definitions are slightly different from ours.
\begin{definition}[Normalised joint entropy~\cite{Hassibi2008The-entropy}]\label{df:normalised}
Let $\{X_{1} , \ldots, X_{n}\}$ be a set of $n$ jointly distributed vector valued Gaussian random variables such that each vector $X_{i}$ is a vector of length $T$. Its \emph{normalised Gaussian entropy function} $g$ is a function in ${\mathbb R}^{2^n}$ such that
\[
g(\alpha) \triangleq \frac{1}{T} h(X_\alpha).
\]

\end{definition}

The only difference between Definitions \ref{df:Gaussian} and \ref{df:normalised} is the scalar multiplier $1/T$. Hence, a normalised Gaussian entropy function must be contained in the set $\mathbb W (\rgv)$. In one sense, our proposed definition is slightly more general as we do not require all the random vectors $X_{i}$ to have the same length. On the other hand, the ``normalising factor'' $1/T$ in Definition \ref{df:normalised} can lead to some interesting results. For example, while we cannot prove that the closure of $  \mathbb W (\rgs) $  is  closed and convex, \cite{Hassibi2008The-entropy} proved that the closure of the set of all normalised Gaussian entropy functions is indeed  closed and convex.
 
\begin{proposition}
Let $\Upsilon^*_{N,n}$\footnote{The subscript $N$ is a mnemonic for the word ``normalised''.} be the set of all normalised Gaussian entropy functions. Then 
\[
\con(\Upsilon^*_{N,n}) = \con(\rgv).
\]
\end{proposition}
\begin{proof}
It can be directly verified from definitions that $\con(\Upsilon^*_{N,n}) \subseteq \con(\rgv)$.
Now, consider any $g\in \rgv$. Then by definition, there exists 
$n$   Gaussian random vectors  $\{ X_{1} , \ldots, X_{n} \}$ 
such that 
\begin{align}
g(\alpha) = h(X_\alpha)
\end{align}
for all $\alpha \subseteq \N$. 
Let $\ell_{i}$ be the length of the random vector $X_{i}$. Assume without loss of generality that 
$\ell_{1} \ge \ell_{i}$ for all $i$.

Let $k= \sum_{i=1}^{n} (\ell_{1} - \ell_{i})$ and  $Y_{1} , \ldots, Y_{k}$ be a set of scalar Gaussian random variables with identity covariance matrix and independent of $\{ X_{1} , \ldots, X_{n} \}$. 
For each $i=1, \ldots, n$, let $r_{i} = \sum_{j=1}^{i} (\ell_{1} - \ell_{i})$ and 
\[
Z_{i} = 
\begin{cases}
X_{i} & \text{ if } \ell_{i} =  \ell_{1} \\
(X_{i},  Y_{r_{i}+1} , \ldots, Y_{r_{i+1}   }) & \text{ otherwise. }
\end{cases}
\]
Clearly, each $Z_{i}$ is a Gaussian vector with the same length $\ell_{1}$. 
Let $g^{*}$ be the normalised entropy function induced by $\{Z_{1},\ldots, Z_{n}\}$. 
It is easy to verify that $\ell_{1} g^{*} = g$. Consequently, $\rgv \subseteq \con(\Upsilon^*_{N,n})$ and the proposition thus follows. 
 \end{proof}

\begin{remark}
Our Proposition \ref{prop:Nis3} can also be derived from \cite[Theorem 5]{Hassibi2008The-entropy}, which proved that for any $g\in \rgv$ when $n =3$, there exists a $\theta^{*} > 0$ such that for all $\theta \ge \theta^{*}$,  $\frac{1}{\theta} g$ is vector Gaussian. However, their proof techniques are completely different.
\end{remark}

\section{Conclusion}
In this paper, we took an information theoretic approach to study determinant inequalities
for positive definite matrices. We showed that characterising all such inequalities for an $n\times n$ positive definite matrix is equivalent to characterising the set of all scalar Gaussian entropy functions for $n$ random variables. While a complete and explicit characterisation of the set is still missing, we obtained  inner and outer bounds respectively by means of linearly representable functions and vector Gaussian entropy functions.  

It turns out that for $n \le 3$, the set of all scalar Gaussian entropy functions is the same as the set of all differential entropy functions.  The latter set is completely characterized by Shannon-type information inequalities.  Consequently, the aforementioned inner and outer bounds agree with each other.  For $n \ge 4$, we showed the contrary, and the problem is seeming very difficult.

%
 

\begin{thebibliography}{10}
\providecommand{\url}[1]{#1}
\csname url@samestyle\endcsname
\providecommand{\newblock}{\relax}
\providecommand{\bibinfo}[2]{#2}
\providecommand{\BIBentrySTDinterwordspacing}{\spaceskip=0pt\relax}
\providecommand{\BIBentryALTinterwordstretchfactor}{4}
\providecommand{\BIBentryALTinterwordspacing}{\spaceskip=\fontdimen2\font plus
\BIBentryALTinterwordstretchfactor\fontdimen3\font minus
  \fontdimen4\font\relax}
\providecommand{\BIBforeignlanguage}[2]{{%
\expandafter\ifx\csname l@#1\endcsname\relax
\typeout{** WARNING: IEEEtran.bst: No hyphenation pattern has been}%
\typeout{** loaded for the language `#1'. Using the pattern for}%
\typeout{** the default language instead.}%
\else
\language=\csname l@#1\endcsname
\fi
#2}}
\providecommand{\BIBdecl}{\relax}
\BIBdecl


\bibitem{CoverBook} T. Cover and J. Thomas, Elements of information theory,  Wiley-Interscience, New York, NY, USA, 1991. ISBN 0-471-06259-6. 

\bibitem{Cover88}
T. Cover,  ``Determinant inequalities via information theory,''  \emph{SIAM. J. Matrix Anal. \& Appl.}, 9(3), pp.384-392, 1988.  

\bibitem{R.W.Yeung1997A-framework}
R.~W. Yeung, ``A framework for linear information inequalities,'' \emph{IEEE
  Trans. Inform. Theory}, vol.~43, no.~6, pp. 1924--1934, Nov 1997.
  
\bibitem{Yeung08}
R.~W.~Yeung, {\em Information Theory and Network Coding}, Springer 2008.

\bibitem{Han78}
T. S. Han, ``Nonnegative entropy measures of multivariate symmetric
correlations,''  \emph{Inform. Contr.}, 36: 133-156, 1978.

\bibitem{Chan2003Balanced}
T.~H. Chan, ``Balanced information inequalities,'' \emph{IEEE Trans. Inform.
  Theory}, vol.~49, pp. 3261 -- 3267, 2003.

\bibitem{Ren70}
A.~R\'enyi, \emph{Probability Theory}.\hskip 1em plus 0.5em minus 0.4em\relax
  Budapest, Hungary: North Holland -- Academiai Kiado, 1970.

\bibitem{CovTho06}
T.~M. Cover and J.~A. Thomas, \emph{Elements of Information Theory},
  2nd~ed.\hskip 1em plus 0.5em minus 0.4em\relax Wiley, 2006.

\bibitem{Wiboonton2010Bijections}
K.~Wiboonton, ``Bijections from $\mathbb R^n$ to $\mathbb R^m$,'' available online at 
  \emph{https://www.math.lsu.edu/~kwiboo1/talkpaper.pdf}, 2010.

\bibitem{Hassibi2008The-entropy}
B.~Hassibi and S.~Shadbakht, ``The entropy region for three gaussian random
  variables,'' in \emph{Information Theory, 2008. ISIT 2008. IEEE International
  Symposium on}, july 2008, pp. 2634 --2638.

\bibitem{Guille2009The-minimal}
L.~Guille, T.~H. Chan, and A.~Grant, ``The minimal set of ingleton
  inequalities,'' \emph{accepted for publications in IEEE Trans. on Inform.
  Theory}, 2011.

\bibitem{ingleton71}
A.~W. Ingleton, ``Representation of matroids.''\hskip 1em plus 0.5em minus
  0.4em\relax London: Academic Press, 1971, pp. 149--167.

\bibitem{Dougherty2009Linear}
R.~Dougherty, C.~Freiling, and K.~Zeger, ``Linear rank inequalities on five or
  more variables,'' \emph{Arxiv preprint cs.IT/0910.0284v3}, 2009.

\bibitem{Kinser2010New-inequalities}
R.~Kinser, ``New inequalities for subspace arrangements,'' \emph{J. Combin.
  Theory Ser. A}, 2010.

\bibitem{Chan2011Truncation}
T.~Chan, A.~Grant, and D.~Pfl{\"u}ger, ``Truncation technique for
  characterising linear polymatroids,'' \emph{accepted for publications in IEEE
  Trans. on Inform. Theory}, 2011.

\bibitem{Zhang1998On-the-characterization}
Z.~Zhang and R.~W. Yeung, ``On the characterization of entropy function via
  information inequalities,'' \emph{IEEE Trans. Inform. Theory}, vol.~44, pp.
  pp. 1440--1452, 1998.

\end{thebibliography}


\end{document}